\documentclass[a4paper,UKenglish,autoref,thm-restate]{lipics-v2021}
\usepackage{xspace}

\hideLIPIcs
\hypersetup{colorlinks,linkcolor=blue,urlcolor=blue,citecolor=blue}

\title{A Nearly Optimal Deterministic Online Algorithm for Non-Metric Facility Location}
\titlerunning{A Nearly Optimal Deterministic Online Algorithm for Non-Metric Facility Location}

\author{Marcin Bienkowski}{Institute of Computer Science, University of Wroc{\l}aw, Poland}{marcin.bienkowski@cs.uni.wroc.pl}{https://orcid.org/0000-0002-2453-7772}{}
\author{Björn Feldkord}{Heinz Nixdorf Institut \& Department of Computer Science, Paderborn University, Germany}{bjoernf@hni.upb.de}{https://orcid.org/0000-0001-6591-2420}{}
\author{Paweł Schmidt}{Institute of Computer Science, University of Wroc{\l}aw, Poland}{pawel.schmidt@cs.uni.wroc.pl}{}{}
\authorrunning{M.\, Bienkowski, B.\, Feldkord, P.\, Schmidt}
\Copyright{Marcin Bienkowski, Björn Feldkord and Paweł Schmidt}


\funding{Supported by Polish National Science Centre grant 2016/22/E/ST6/00499 and 
by German Research Foundation (DFG) within the Collaborative Research Centre ``On-The-Fly Computing'' 
under the project number 160364472 --- SFB 901/3.}

\nolinenumbers

\EventEditors{Markus Bl\"{a}ser and Benjamin Monmege}
\EventNoEds{2}
\EventLongTitle{38th International Symposium on Theoretical Aspects of Computer Science (STACS 2021)}
\EventShortTitle{STACS 2021}
\EventAcronym{STACS}
\EventYear{2021}
\EventDate{March 16--19, 2021}
\EventLocation{Saarbr\"{u}cken, Germany (Virtual Conference)}
\EventLogo{}
\SeriesVolume{187}
\ArticleNo{48}

\newcommand{\F}[1]{F_{#1}}
\newcommand{\E}{\mathbf{E}}
\newcommand{\e}{\mathrm{e}}
\newcommand{\cost}{\mathsf{cost}} 
\newcommand{\maxcost}{\rho}
\newcommand{\set}{\mathsf{set}}
\newcommand{\primal}{\ensuremath{\mathcal{P}}\xspace}
\newcommand{\dual}{\ensuremath{\mathcal{D}}\xspace}
\newcommand{\ALG}{\textsc{Alg}\xspace}
\newcommand{\OPT}{\textsc{Opt}\xspace}
\newcommand{\INT}{\textsc{Int}\xspace}
\newcommand{\INTFAC}{\textsc{IntFac}\xspace}
\newcommand{\DET}{\textsc{Det}\xspace}
\newcommand{\FRAC}{\textsc{Frac}\xspace}
\newcommand{\tG}{\tilde{G}}
\newcommand{\hy}{\hat{y}}
\newcommand{\hY}{\hat{Y}}
\newcommand{\poly}{\mathrm{poly}}

\begin{document}
\maketitle

\begin{abstract}
In the online non-metric variant of the facility location problem, there is a
given graph consisting of a set $F$ of facilities (each with a certain opening
cost), a set $C$ of potential clients, and weighted connections between them.
The online part of the input is a sequence of clients from~$C$, and in response
to any requested client, an online algorithm may open an additional subset of
facilities and must connect the given client to an open facility. 

We give an online, polynomial-time deterministic algorithm for this problem,
with a competitive ratio of $O(\log |F| \cdot (\log |C| + \log \log |F|))$. The
result is optimal up to loglog factors. Our algorithm improves over the $O((\log
|C| + \log |F|) \cdot (\log |C| + \log \log |F|))$-competitive construction that
first reduces the facility location instance to a set cover one and then later
solves such instance using the deterministic algorithm by Alon et al.~[TALG
2006]. This is an asymptotic improvement in a typical scenario where $|F| \ll
|C|$.

We achieve this by a more direct approach: we design an algorithm for a
fractional relaxation of the non-metric facility location problem with clustered
facilities. To handle the constraints of such non-covering LP, we combine the
dual fitting and multiplicative weight updates approach. By maintaining certain
additional monotonicity properties of the created fractional solution, we can
handle the dependencies between facilities and connections in a~rounding
routine.

Our result, combined with the algorithm by Naor et al.~[FOCS 2011] yields the
first deterministic algorithm for the online node-weighted Steiner tree problem.
The resulting competitive ratio is $O(\log k \cdot \log^2 \ell)$ on graphs of
$\ell$ nodes and $k$ terminals.
\end{abstract}

\section{Introduction}

The facility location (FL) problem~\cite{AaByMa16} is one of the best-known
examples of network design problems, extensively studied both in operations
research and in computer science. Its simple definition, \textsf{NP}-hardness,
and rich combinatorial structure have led to developments of tools and solutions
in key areas of approximation algorithms, combinatorial optimization, and linear
programming. 

An instance of the FL problem consists of a set $F$ of facilities, each with a
certain opening cost, and a set $C$ of clients. $F$ and $C$ can be seen as two
sides of a bipartite graph. The undirected edges between them have lengths that
can either satisfy the triangle inequality (\emph{metric FL}) or be arbitrary
(\emph{non-metric FL}). The goal is to open a subset of facilities and connect
each client to an open facility. The total cost (the sum of opening and
connection costs) is subject to minimization. In the metric scenario, by taking
a metric closure, one can assume that each facility is reachable by each client,
but it is not the case for the non-metric variant.

\subparagraph*{Instances and Objectives.}

In this paper, we focus on an online variant of the non-metric FL problem. We
first formalize the offline variant in a way that makes a connection to the
online variant more apparent.

A \emph{facility-client graph} $G = (F,C,E,\cost)$ is a bipartite graph, whose
one side is the set~$F$ of facilities and another side is the set of clients
$C$. Set $E \subseteq F \times C$ contains available facility-client
connections (edges). We use function $\cost$ to denote both costs of opening
facilities and connection costs (edge lengths). All costs are non-negative.

An instance of the non-metric FL problem is a pair $(G,A)$, where $G = (F,C,E,\cost)$ is
a~facility-client graph and $A \subseteq C$ is a subset of \emph{active}
clients. A feasible solution to such instance is a~set of open (purchased)
facilities $F' \subseteq F$ and a subset of purchased edges~$E' \subseteq E$,
such that any active client $c \in A$ is connected by a purchased edge to an
open facility. The cost of such solution is equal to $\sum_{f \in F'}
\cost(f) + \sum_{e \in E'} \cost(e)$.

For any facility-client graph $G$, we define its aspect ratio $\Delta_G$ as the
ratio of the largest to smallest positive cost in $G$. These costs include both
facilities and connection costs.\footnote{In the standard definition of the
aspect ratio, only distances are taken into account.} Note that the aspect ratio
is a property of $G$ and is independent of the set of active clients $A$.

\subparagraph*{Online Scenario.}

In an \emph{online variant} of the FL problem, the facility-client graph $G$ is
known in advance, but neither elements of $A$ nor its cardinality are known
up-front by an~online algorithm \ALG. The clients from $A$ appear one by one.
Upon seeing a new active client, \ALG may purchase additional facilities and
edges, with the requirement that facilities and edges purchased so far must
constitute a feasible solution to all presented active clients. The total cost
of \ALG is denoted by $\ALG(G,A)$. (We sometimes use $\ALG(G,A)$ to also denote
the solution computed by \ALG.) Purchase decisions are final and cannot be
revoked later. The goal is to minimize \emph{the competitive ratio}, defined as
$\sup_{(G,A)} \{\ALG(G,A) / \OPT(G,A) \}$, where \OPT is the optimal (offline)
algorithm.


\subsection{Related Work}

Most of the prior work has been devoted to the offline scenario. While the
metric variant of the FL problem admits O(1)-approximation
algorithms~\cite{ByrAar10,ChaGuh05,ChuShm03,JaMMSV03,JaMaSa02,KoPlRa00,Li13,MaYeZh06,ShTaAa97},
the best approximation ratio for the non-metric one is $O(\log
|A|)$~\cite{Hochba82}, and it cannot be asymptotically improved unless
$\textsf{NP} \subseteq \textsf{DTIME}[n^{O(\log \log n)}]$~\cite{Feige98}. For a
more comprehensive treatment of the offline scenario, including a multitude of
variants, we refer the reader to the entry in the Encyclopedia of
Algorithms~\cite{AaByMa16} or the survey by Shmoys~\cite{Shmoys00}.

For the online metric FL, the problem was resolved over ten years ago by
Meyerson~\cite{Meyers01} and Fotakis~\cite{Fotaki08}: the lower and upper bounds
on the competitive ratio are $\Theta(\log |A| / \log \log |A|)$, both for
deterministic and randomized algorithms. Simpler deterministic algorithms
attaining slightly worse competitive ratio of $O(\log |A|)$ were given by
Anagnostopoulos et al.~\cite{AnBeUH04} and Fotakis~\cite{Fotaki07}. Note that
the optimal competitive ratio in the metric case is independent of the set $C$ of potential clients.


\subsection{Previous Work on Online Non-Metric Facility Location}

For the non-metric FL, the first and currently best online algorithm was a
\emph{randomized} algorithm by Alon et al.~\cite{AlAABN06}. It achieves the
competitive ratio of $O(\log |F| \cdot \log |A|)$. It is based on solving
a~natural fractional relaxation of the problem: there is a fractional
\emph{opening variable}~$y_f$ for each facility $f$ and a \emph{connection
variable} $x_{c,f}$ for a client $c$ and a \emph{covering} facility $f$ (facility
to which $c$ could be connected). Once a~client~$c$ arrives, for each covering facility $f$
independently, their algorithm increases either $y_f$ or~$x_{c,f}$, whichever is
smaller, using multiplicative update method (see, e.g.,~\cite{ArHaKa12}). The
client $c$ is considered fractionally served once the sum of terms $\min \{
x_{c,f},\, y_f \}$ over all covering facilities is at least $1$. The resulting
competitive ratio is $O(\log |F|)$. 

The computed fractional solution can be then rounded using a random threshold 
$\theta_f$ common for an opening variable $y_f$ and all connection
variables involving facility $f$. Once any variable exceeds its threshold, it is
rounded up to $1$ and the corresponding object (facility or connection) is
purchased. Dynamically adjusting $\theta_f$ to have expectation $\Theta(1/\log |A|)$
guarantees that the resulting integral solution is feasible with high
probability and the rounding part incurs a~factor of $O(\log |A|)$ in the
competitive ratio.

To the best of our knowledge, no non-trivial deterministic algorithm was
published so far. In particular, the online network design problems (including
the non-metric FL problem) have been listed as unresolved challenges by
Buchbinder and Naor~\cite[Section 1.1]{BucNao09b}. That said, the non-metric
facility location can be reduced to a set cover. A usable reduction (not
inducing an exponential blow-up of the input size) was given by Kolen and
Tamir~\cite{KolTam90}: it preserves the solution costs up to constant factors
and creates a set cover instance consisting of $m = \Theta(|F| + |C| \cdot \log
\Delta_G)$ sets and $n = \Theta(|C| \cdot \log \Delta_G)$ elements. Using doubling 
techniques described in \autoref{sec:doubling}, one could assume that $\Delta_G = O(|F|
\cdot |C|)$. Applying the deterministic algorithm for the online set cover problem by
Alon et al.~\cite{AlAABN09} yields a~solution whose competitive ratio
is $O(\log m \cdot \log n) = O((\log|C| + \log |F|) \cdot (\log |C| + \log \log |F|))$.


\subsection{Our Result}

In our paper, we improve the bound above, replacing the first factor of $O(\log
|C| + \log |F|)$ by $O(\log |F|)$. This is an asymptotic improvement in a
typical scenario where $|F| \ll |C|$.

\begin{restatable}[]{theorem}{maintheorem}
\label{thm:main_result}
There exists a deterministic polynomial-time $O(\log |F| \cdot (\log |C| + \log
\log |F|))$-competitive algorithm for the online non-metric facility
location problem on set~$F$ of facilities and set $C$ of clients.
\end{restatable}

Our algorithm attains a nearly optimal competitive ratio, as no deterministic
algorithm can have a ratio smaller than $\Omega(\log |F| \cdot \log |C| /(\log
\log |F| + \log \log |C|))$. This follows by the lower bound for the online set
cover problem~\cite{AlAABN06,AlAABN09} and holds even for uniform facility
costs. If we restrict our attention to the \emph{polynomial-time} deterministic
solution, then a stricter lower bound of $\Omega(\log |F| \cdot \log |C|)$ holds
(assuming $\textsf{BPP} \neq \textsf{NP}$)~\cite{Korman04}.


\subparagraph*{Challenges.}

The description of the randomized algorithm by Alon et al.~\cite{AlAABN06} given
above seems deceptively simple, but it hides an important and subtle property,
implicitly exploited by the authors. Namely, the threshold $\theta_f$ is common
for facility~$f$ and all connections to it. This ensures the necessary
dependency: once $\min \{ x_{c,f},\, y_f \} \geq \theta_f$, the rounding
purchases \emph{both facility $f$ and a connection from $c$ to $f$}. (Note that
the left-hand side of this inequality is the amount that their fractional
solution controls.)

It is unclear how to directly extend this property to deterministic rounding. A
straightforward attempt would be to focus on facilities only and round them in
a deterministic fashion ensuring the necessary coverage of each client. However,
neglecting the connection costs in the rounding process easily leads to a
situation, where the facilities are rounded ``correctly'', but the cost of
connecting a client to the closest open facility in the integral solution is
incomparably larger than the corresponding fractional cost. 

We note that all known deterministic schemes that round fractional solutions
generated by the multiplicative updates operate in rather limited scenarios,
where elements have to be covered or packed and all important interactions
between elements are handled at the time of constructing the fractional solution.
This is the case for the deterministic rounding for the set cover
problem~\cite{AlAABN09,BucNao09} and the throughput-competitive virtual
circuit routing problem~\cite{AwAzPl93,BucNao09b}. These methods are based on
derandomizing the method of pessimistic estimators~\cite{Raghav88} in an~online
manner, by transforming a pessimistic estimator into a potential
function~\cite{Young95} that can be controlled by the deterministic rounding
process.

\subparagraph*{Our Techniques.}

In our solution, we create a new linear relaxation of the problem. We first
round the graph distances to powers of $2$. For any client, we cluster
facilities that have the same distance to this client. (Note that such clusters are
client-dependent.) To solve the fractional variant, we run two schemes in
parallel: we increase connection variables $x_{c,t}$ corresponding to clusters
at distance~$t$, and increase facility variables $y_f$ for all facilities in
``reachable'' clusters (where the corresponding connection variables are $1$). The
increases in these variables use two different frameworks: dual fitting for
linear increases of connection variables and a primal-dual scheme involving
multiplicative updates for facility variables. Ensuring an appropriate balance
between these two different types of updates is one of the technical
difficulties that we tackle in this paper.

We stop increasing variables once there exists a collection of clusters that are
both ``fractionally open'' (sum of variables $y_f$ within these clusters is
$\Omega(1)$) and ``reachable'' by the considered client. To argue about the
existence of such a collection, we use both LP inequalities and structural
properties of our fractional algorithm. 

Finally, we construct a deterministic rounding routine. We focus on facilities
only, neglecting whether particular clients are active or not and how far they
are from a given facility. However, we strengthen rounding properties, ensuring,
for (some) collections of clusters, that if the sum of opening variables in
these collections is $\Omega(1)$, then the integral solution contains an open
facility in one of these clusters. This ensures that, for a considered client
$c$, the integral solution contains a facility whose distance from $c$ is
asymptotically not larger than the cost invested for connecting $c$ in the
fractional solution. Ultimately, this yields the desired dependency between
facilities and connections.

\subparagraph*{Note about Up-Front Knowledge of the Facility-Client Graph.}

Unlike for the randomized variant, obtaining sub-linear guarantees for a
deterministic solution requires knowing a priori the set of potential
client-facility connections. To see this, consider a graph of $|F|$ facilities
with unit opening costs and the set of $|C| = |F|$ clients. The graph edges are
constructed dynamically as clients are activated and all revealed possible
connections are of cost $0$. The first active client can be connected to all
facilities. Each subsequent client can be connected to all facilities but the
ones already open by an algorithm. This way an online algorithm needs to
eventually open all facilities, for a~total cost of $|F|$. On the other hand,
the offline optimal algorithm can open the last facility opened by an online
algorithm and connect all clients to this facility paying just $1$. Thus, under
the unknown-graph assumption, the competitive ratio of any deterministic
algorithm would be at least $|F|$.


\subsection{Preliminaries and Paper Organization}

Let $T_G$ contain all powers of two between the largest and the smallest
positive distance (inclusively) and also number $0$. In particular, $T_G$
contains all distances in $G$ and $|T_G| \leq 2 + \log \Delta_G$. Whenever $G$
is clear from the context, we drop the $G$ subscript. 

We may assume that $F$ contains at least two facilities and $C$~contains at
least two clients, as otherwise the problem becomes trivial. For a facility $f
\in F$, let $\set(f)$ be the set of clients that may be connected to $f$. For
any client~$c \in C$ and distance $t \in T$, cluster $F_{c,t}$ contains all
facilities that are incident to $c$ using edges of cost $t$. Note that for a
fixed $c$, clusters $\F{c,t}$ are disjoint (no client has two connections of
different costs to the same facility).

\subparagraph*{Powers-of-Two Assumption.}

In the whole paper, we assume that all facilities and connection costs are
either equal to $0$ or are powers of $2$ and are at least $1$. This can be
easily achieved by initial scaling of  positive costs and distances, so that
they are at least $1$ and rounding positive ones up to the nearest power of two.
This transformation changes the competitive ratio at most by a~factor of~$2$.

\subparagraph*{Paper Overview.}

Our core approach is to solve a~carefully crafted fractional relaxation of the
problem (\autoref{sec:fractional}), and then round it in a deterministic
fashion (\autoref{sec:rounding}). This way, we obtain a deterministic online
algorithm \INT that on any input $(G = (F,C,E,\cost),A)$ computes a feasible
solution of cost 
\[ 
	\INT(G,A) \leq O(\log |F| \cdot (\log |C| + \log \log \Delta_G))
	\cdot \OPT(G,A) + 2 \cdot \max_{f \in F} \cost(f).
\]
Moreover \INT runs in time $\poly(|G|, |A|, \max_{e \in E} \cost(e), \max_{f \in
F} \cost(f))$. In \autoref{sec:doubling}, we apply doubling and edge pruning
techniques, to get rid of dependencies on costs in the running time and on
$\Delta_G$ in the competitive ratio, achieving guarantees of
\autoref{thm:main_result}.

\subparagraph*{Application to Node-Weighted Steiner Tree.}

Our result has an immediate application to the online node-weighted Steiner
tree (NWST) problem. Namely, when we combine \autoref{thm:main_result} 
with the randomized solution for 
NWST by Naor et al.~\cite{NaPaSi11}, we obtain the first deterministic algorithm
with polylogarithmic competitive ratio (see \autoref{sec:steiner}).


\section{Fractional Solution}
\label{sec:fractional}

We fix an instance $(G = (F,C,E,\cost),A)$ of the online non-metric facility
problem. For each facility $f$, we introduce an \emph{opening} variable $y_f
\geq 0$ (fractional opening of $f$) and for each client~$c$ and each distance $t
\in T$ a~\emph{connection} variable $x_{c,t} \geq 0$. Intuitively, $x_{c,t}$
denotes how much, fractionally, client $c$ invests into connections to
facilities from cluster $\F{c,t}$. For any set $F'$ of facilities we use $y(F')$
as a~shorthand for $\sum_{f \in F'} y_f$.

\subparagraph*{Primal Program.}
After $k$ clients from $A$ arrive (we denote their set by $A_k$), we consider
the following linear program~$\primal_k$. 
\begin{align*}
    \text{minimize \quad} & 
  	\sum_{f \in F} \cost(f) \cdot y_f +  \sum_{c \in A_k} \sum_{t \in T} t \cdot x_{c,t} && \\
    \text{subject to \quad} 
        & x_{c,t} \geq z_{c,t} && \text{for all $c \in A_k,\, t \in T$}, \\
        & y(\F{c,t})   \geq z_{c,t} && \text{for all $c \in A_k,\, t \in T,$} \\
            & \sum_{t \in T} z_{c,t} \geq 1
			&& \text{for all $c \in A_k$},
\end{align*}
and non-negativity of all variables.

\subparagraph*{Serving Constraints.}
The LP constraints combined are equivalent to the set of the following
(non-linear) requirements
\begin{align}
\label{eq:lp_non_linear}
    \sum_{t \in T} \min \left\{ x_{c,t},\; y(\F{c,t}) \right\} \geq 1
    		&& \text{for all $c \in A_k$}.
\end{align}
We call \eqref{eq:lp_non_linear} for client $c$ the \emph{serving constraint}
for client $c$. In our description, we omit variables~$z_{c,t}$ and the original
constraints, ensuring only that the serving constraints hold and implicitly
setting $z_{c,t} = \min \{ x_{c,t},\; y(\F{c,t}) \}$.

The LP above is indeed a valid relaxation of the FL problem. To see this, take
any feasible integral solution. For any facility $f$ opened in the integral
solution, set variable $y_f$ to $1$. For each client $c$ connected to facility
$f$, set variable $x_{c,\tau}$ to $1$, where $\tau = \cost(f,c)$. This
guarantees that $\min \{x_{c,\tau},\, y( F_{c,\tau}) \} = 1$, and thus the
serving constraint \eqref{eq:lp_non_linear} is satisfied for each client $c$. 

\subparagraph*{Dual Program.}
The program $\dual_k$ dual to $\primal_k$ is 
\begin{align*}
    \text{maximize \quad} & 
  	\sum_{c \in A_k} \gamma_c && \\
    \text{subject to \quad} 
        & \gamma_c \leq \alpha_{c,t} + \beta_{c,t} && \text{for all $c \in A_k,\, t \in T$}, \\
        & \alpha_{c,t} \leq t && \text{for all $c \in A_k,\, t \in T$}, \\
        & \sum_{c \in \set(f) \,\cap\, A_k} \beta_{c,\,\cost(f,c)} \leq \cost(f) && \text{for all $f \in F$},
\end{align*}
and non-negativity of all variables.


\subsection{Overview}
\label{sec:frac_overview}

Our algorithm \FRAC creates a solution to $\primal_k$, ensuring that the serving
constraint \eqref{eq:lp_non_linear} holds for all clients~$c \in A_k$. As
outlined in the introduction, the computed solution guarantees some additional
properties that are useful for the rounding part later.

Whenever a client $c$ arrives, \FRAC increases connection variables $x_{c,t}$
one by one starting from the smallest $t$, at the pace proportional to $1/t$. We
ensure that $x_{c,t} \in [0,1]$, i.e., once any of these variables reaches $1$,
\FRAC stops increasing them. A distance $t$, for which $x_{c,t} = 1$, is called
\emph{saturated}.

In parallel to manipulating variables $x_{c,t}$, \FRAC increases all variables
$y_f$ for facilities reachable from client $c$ using saturated distances. The
variables $y_f$ are increased using the multiplicative update rule~\cite{ArHaKa12} (scaled
appropriately to take costs of facilities into account). 

Together with the solution to $\primal_k$, \FRAC also constructs an
\emph{almost-feasible} solution to~$\dual_k$. That is, its solution to $\dual_k$
is feasible when all dual variables are scaled down by a factor of $O(\log
|F|)$. By the weak duality, the scaled-down value of this solution serves as a
lower-bound for the optimum. Thus, as typical for the primal-dual type of
analysis, the dual variables can be thought of as budgets whose increase
balances the increase of primal variables.


\subsection{Algorithm FRAC}

At the very beginning, before any client arrives, $\FRAC$ sets all variables
$y_f$ to $0$ for all positive-cost facilities and to $1$ for zero-cost ones.
There are no other variables as the set~$A_0$ of active clients is empty. Note
that the dual program already contains the last type of constraints, but the
sums on their left-hand sides range over empty sets of $\beta$ variables, and
hence these constraints are trivially satisfied.

Whenever a new client $c$ arrives in step $k$, \FRAC updates the primal (dual)
programs from $\primal_{k-1}$ ($\dual_{k-1}$) to $\primal_k$ ($\dual_k$), and
then computes a feasible solution to $\primal_k$ (based on the already
created solution to $\primal_{k-1}$) and a nearly-feasible solution to
$\dual_k$.

\begin{description} 
    \item[New variables in primal and dual programs:] 
        \FRAC sets $x_{c,t} \gets 0$ for all $t \in T \setminus
        \{0\}$ and sets $x_{c,0} \gets 1$. In the dual solution, it sets 
        $\gamma_c \gets 0$, 
        $\alpha_{c,t} \gets 0$ and $\beta_{c,t} \gets 0$ for all $t \in T$.

    \item[Update primal program:] A new serving constraint $\sum_{t \in T} \min
        \{ x_{c,t}, \, y(\F{c,t}) \} \geq 1$ appears in the primal program (and
        is violated unless $y(\F{c,0}) \geq 1$). As we never decrease primal
        variables, the serving constraints \eqref{eq:lp_non_linear} that existed
        already in $\primal_{k-1}$ are satisfied and will not become
        violated.
        
    \item[Update dual program:] New constraints appear in the dual program and
    new variables $\beta_{c,t}$ appear on the left-hand side of the already
    existing inequalities.  Since the new variables are initialized to $0$, the
    validity of all dual constraints is unaffected.
    
    \item[Update primal and dual solutions:] 
    Let $T_c^{1} = \{ t \in T : x_{c,t} \geq 1 \}$ be the set 
    of saturated distances, i.e., initially \FRAC sets $T_c^1 \gets \{ 0 \}$.
    While the serving constraint for~$c$ is violated, \FRAC executes the 
    \emph{update operation} consisting of the following steps:
    \begin{enumerate} 
        \item Set $\gamma_c \gets \gamma_c +1$.
        \item For each $t\in T$, independently, adjust one dual
        variable: if $t \in T_c^1$, then set $\beta_{c,t} \gets \beta_{c,t} + 1$ and 
        otherwise set $\alpha_{c,t} \gets \alpha_{c,t} + 1$.
        \item If $T_c^1 \subsetneq T$,   
            choose \emph{active distance} $t^* \gets \min (T \setminus T_c^1)$ 
            to be the smallest non-saturated distance,
            and then set $x_{c,t^*} \gets x_{c,t^*} + 1/t^*$.
            (Note that $0 \in T_c^1$, and thus $t^* > 0$.)
        \item 
        \label{step:augmentation_step}
        For any facility $f \in \biguplus_{t \in T_c^1} \F{c,t}$, independently, 
        perform \emph{augmentation of $y_f$}, setting
            \[
                y_f \gets \left (1 + \frac{1}{\cost(f)} \right) \cdot y_f + \frac{1}{|F|\cdot \cost(f)}  .
            \]
        \item Update the set of saturated distances, setting $T_c^1 \gets \{ t \in  T : x_{c,t} \geq 1 \}$.
    \end{enumerate}
\end{description}
We now argue that if variable $y_f$ is augmented in Step 4, then $\cost(f) > 0$
(i.e., Step 4 is well defined). Let $\tau = \cost(c,f)$. As $y_f$ is augmented,
the distance $\tau$ is saturated ($x_{c,\tau} = 1$). If $\cost(f) = 0$, then
$y_f$ would have been initialized to $1$, and then $y(F_{c,\tau}) \geq 1$, in
which case the serving constraint for $c$ would be already satisfied.

\subparagraph*{Sidenote about T.} 
For the sake of coherence and more streamlined analysis,
\FRAC increases also connection variables $x_{c,t}$ to empty sets $\F{c,t}$, i.e., 
invests into distances to non-existing facilities. Fixing this overspending would 
not lead to asymptotic improvement of the performance.


\subsection{Structural Properties}

We focus on a single client $c$ processed by \FRAC. We start with a property of
connection variables~$x_{c,t}$. The distances from $T$ that are neither
saturated nor active are called \emph{inactive}. The following claim follows by an
immediate induction on update operations performed by \FRAC.

\begin{lemma}
\label{lem:structural}
At all times when a client $c$ is considered, $x_{c,t} \in [0,1]$ for any $t \in T$.
In particular, $x_{c,t} = 1$ for any saturated distance $t \in T_c^1$. 
Furthermore, 
\begin{enumerate}
\item either all distances are saturated, 
\item or there exists 
an active distance $t^* > 0$, such that (i) all smaller distances are saturated, 
and (ii) all larger distances are inactive and the corresponding variables $x_{c,t}$ 
are equal to zero.
\end{enumerate}
Augmentation is performed on variables $y_f$ corresponding to facilities 
whose distance from $c$ is saturated.
\end{lemma}

\begin{lemma}
\label{lem:primal_is_feasible}
On any input $(G = (F,C,E,\cost), A)$, \FRAC returns a feasible solution and
runs in time $\poly(|G|, |A|, \max_{e \in E} \cost(e), \max_{f \in F}
\cost(f))$.
\end{lemma}

\begin{proof}
Fix any client $c \in A$. By the definition of $\FRAC$, it takes $t$ update
operations to increase value $x_{c,t}$ from $0$ to $1$. Hence, after $\sum_{t \in
T} t < 2 \cdot \max_{e \in E} \cost(e)$ update operations, all connection
variables are equal to $1$. From that point on, all variables $y_f$ for $f \in
\biguplus_{t \in T} \F{c,t}$ are augmented in each update operation. Each
variable $y_f$ can be augmented at most $|F| \cdot \cost(f)$ times till it
reaches or exceeds~$1$. That is, after at most $2 \cdot \max_{e \in E} \cost(e)
+ |F| \cdot \max_{f \in F} \cost(f)$ update operations, the serving constraint
is satisfied, i.e., the generated solution is feasible.
\end{proof}


The following lemma shows the crucial property of \FRAC. Namely for any client
$c$, there exist a~``good'' distance $\tau$, such that the collection of
clusters $F_{c,t}$ at distance $t \leq \tau$ is together fractionally half-open
and that \FRAC invested $\Omega(\tau)$ into connecting client $c$. For any
client~$c$ and distance $t \in T$, we define a set $S_{c,t}$ to be a collection
of clusters alluded to in the introduction. 
\[
    S_{c,t} = \biguplus_{t' \in T \,:\, t' \leq t} F_{c,t'}
\]

\begin{lemma}
\label{lem:good_distance}
Once \FRAC finishes serving client $c$, there exists a distance $\tau \in T$,
such that $y(S_{c,\tau}) \geq 1/2$ and $\sum_{t \in T} t \cdot x_{c,t}  \geq \tau/2$.
\end{lemma}

\begin{proof}
We consider the state of variables once \FRAC finishes serving client $c$. Let
$t^* > 0$ be the largest distance from $T$ for which $x_{c,t^*} > 0$. As the
serving constraint for client $c$ is satisfied, we have
\begin{equation}
    \label{eq:expanded_serving_cond}
	1 \leq \sum_{t \in T} \min \left\{ x_{c,t},\, y(\F{c,t}) \right\} 
     = \min \left\{ x_{c,t^*},\, y(\F{c,t^*}) \right\} 
     + \sum_{t\in T \,:\, t < t^*} \min \left\{ x_{c,t},\, y(\F{c,t}) \right\} .
\end{equation}
We pick $\tau$ depending on the value of the last term
of~\eqref{eq:expanded_serving_cond}. 

If $\min \{ x_{c,t^*},\, y(\F{c,t^*}) \} \geq 1/2$, we set $\tau = t^*$. Then,
$y(S_{c,\tau}) \geq y(\F{c,\tau}) \geq \min \{ x_{c,\tau},\, y(\F{c,\tau}) \} $ $
\geq 1/2$, and the first condition of the lemma follows. Furthermore,
$\sum_{t\in T} t \cdot x_{c,t} \geq \tau \cdot x_{c,\tau} \geq \tau/2$. 

Otherwise, $\min \left\{ x_{c,t^*},\, y(\F{c,t^*}) \right\} < 1/2$, and then, by
\eqref{eq:expanded_serving_cond}, $\sum_{t\in T \,:\, t < t^*} \min \{
x_{c,t},\, y(\F{c,t}) \} \geq 1/2$. In such case, we choose $\tau$ as the
largest distance from~$T$ smaller than~$t^*$. Then 
\[ 
    y(S_{c,\tau}) = \sum_{t\in T \,:\, t \leq \tau} y(\F{c,t}) 
\geq \sum_{t\in T \,:\, t \leq \tau} \min \{ x_{c,t},\, y(\F{c,t}) \} \geq 1/2,
\]
i.e., the first condition of the lemma holds. By \autoref{lem:structural},
either $t^*$ is active at the end of processing~$c$ or all distances become
saturated and $t^*$ is the largest distance from $T$. In either case, $x_{c,t} =
1$ for any distance $t < t^*$, and thus in particular $x_{c,\tau} = 1$. Hence,
the second part of the lemma holds as $\sum_{t\in T} t\cdot x_{c,t} \geq \tau
\cdot x_{c,\tau} = \tau$.
\end{proof}


\subsection{Dual Solution is Almost Feasible}

Using primal-dual analysis, we may show that the generated dual solution
violates each constraint at most by a factor of $O(\log |F|)$. 

\begin{lemma}
\label{lem:log_augmentations}
For any facility $f$, \FRAC augments $y_f$ at most $O(\log |F|) \cdot \cost(f)$ times.
\end{lemma}

\begin{proof}
First, we observe that variable $y_f$ can be augmented only if prior to
augmentation it is smaller than~$1$. To show that, observe that the augmentation
of $y_f$ occurs only when \FRAC processes an active client $c \in \set(f)$. Let
$\tau = \cost(f,c)$, i.e., $f \in F_{c,\tau}$. As \FRAC augments~$y_f$, the
distance $\tau$ must be saturated, i.e., $x_{c,\tau} = 1$. On the other hand,
the serving constraint~\eqref{eq:lp_non_linear} is not satisfied when $y_f$ is
augmented, and thus $\min \{ x_{c,\tau},\, y(\F{c,\tau}) \} < 1$ which implies
that $y_f$ must be strictly smaller than $1$. 

In particular, if $\cost(f) = 0$, then $y_{f}$ is set to $1$ immediately at the
beginning, and hence no augmentation of $y_f$ is ever performed, and the lemma
follows trivially. As all non-zero costs are at least $1$, below we assume
$\cost(f) \geq 1$. 

During the first $\cost(f)$ augmentations, the value of $y_{f}$ increases from
$0$ to at least $1/|F|$ (due to additive increases). Next, during the subsequent
$\lceil \log_{1+1/\cost(f)}|F| \rceil$ augmentations, the value of $y_{f}$
reaches at least $1$ (due to multiplicative increases), and hence it will not be
augmented further. In total, the number of augmentations is upper-bounded by
$\cost(f) + \lceil \log_{1+1/\cost(f)}|F| \rceil = O(\log |F|) \cdot \cost(f)$.
In the last relation, we used $\cost(f) \geq 1$.
\end{proof}

\begin{lemma}
\label{lem:dual_almost_feasible}
\FRAC violates each dual constraint at most by a factor of $O(\log |F|)$.
\end{lemma}

\begin{proof}
We show the claim for all types of constraints in the dual program.
\begin{enumerate}
    
\item 
Each dual constraint $\gamma_c \leq \alpha_{c,t} + \beta_{c,t}$ always
holds with equality as together with $\gamma_c$, for each $t\in T$, \FRAC
increments either $\alpha_{c,t}$ or $\beta_{c,t}$.

\item 
Consider a constraint $\alpha_{c,t} \leq t$. Initially $\alpha_{c,t} = 0$
when client $c$ appears, and it is incremented in an update operation only if
distance $t$ is not saturated. Distances are processed from the smallest to the largest,
and it takes exactly $t'$ update operations for a~distance $t' \in T$ to become
saturated. Therefore, $\alpha_{c,t}$ can be incremented at most $\sum_{t' \in T
: t' \leq t} t'$ times. If $t = 0$, then $\alpha_{c,t} = 0$ trivially.
Otherwise, we use the fact that $T \setminus \{0\}$ contains only powers of $2$,
and hence $\alpha_{c,t} \leq \sum_{t' \in T : t' \leq t} t' < 2 \cdot t$.

\item 
Finally, fix any facility $f^* \in F$ and consider the constraint $\sum_{c \in
\set(f^*) \,\cap\, A_k} \beta_{c,\,\cost(f^*,c)} \leq \cost(f^*)$. We want to
show that this constraint is violated at most by a factor of $O(\log |F|)$, i.e., that
\begin{equation}
\label{eq:dual_feasibility}
    \sum_{c \in \set(f^*) \,\cap\, A_k} \beta_{c,\,\cost(f^*,c)} \leq O(\log |F|) \cdot \cost(f^*) .
\end{equation}
The left-hand side of \eqref{eq:dual_feasibility} is initially $0$ and it is
incremented only when \FRAC processes some active client $c^* \in \set(f^*)$. In
a single update operation, \FRAC may increment multiple $\beta$ variables, but
only one of them, namely $\beta_{c^*,\,\cost(f^*,c^*)}$, contributes to the
growth of the left-hand side of \eqref{eq:dual_feasibility}. If variable
$\beta_{c^*,\,\cost(f^*,c^*)}$ is incremented, it means that the distance $\tau
= \cost(f^*,c^*)$ is already saturated, i.e., $\tau \in T_{c^*}^1$. Thus, in the
same update operation, \FRAC augments all variables $y_f$ for $f \in
\biguplus_{t \in T_{c^*}^1} F_{c^*,t}$. This set of facilities includes cluster
$F_{c^*,\tau}$ and thus also facility $f^*$. By
\autoref{lem:log_augmentations}, the augmentation of $y_{f^*}$ may happen at most
$O(\log |F|) \cdot \cost(f^*)$ times, which implies our claim.
\qedhere
\end{enumerate}
\end{proof}


\subsection{Competitive Ratio of FRAC}

Finally, we show that in each update operation the growth of the primal cost is
at most constant times the growth of the dual cost. This will imply the
competitive ratio of \FRAC.

\begin{lemma}
\label{lem:primal_dual_relation}
For any step $k$, the value of the solution to $\primal_k$ computed by \FRAC is
at most $3$ times the value of its solution to $\dual_k$. 
\end{lemma}

\begin{proof}
As the values of both solutions are initially zero, it suffices to analyze the
growth of the primal and dual objectives for a single update operation. The
value of the dual solution grows by~$1$ as $\gamma_c$ is incremented only for
the requested client $c$. Thus, it is sufficient to show that the primal
solution increases at most by $3$.

By $y_f$, $x_{c,t}$ and $T_c^1$, we understand the values of these variables
before an~update operation. Let $F_1 = \biguplus_{t \in T_c^1} \F{c,t}$. As the
serving constraint for client $c$ is not satisfied at that point,
\begin{equation}
\label{eq:y_f_before_update}
    1 > \sum_{t \in T} \min \left\{ x_{c,t},\, y(\F{c,t}) \right\} 
    \geq \sum_{t \in T_c^1} \min \left\{ x_{c,t},\, y(\F{c,t}) \right\} 
    \geq \sum_{t \in T_c^1} y(\F{c,t})
    = y(F_1) .
\end{equation}
In the last inequality we used that (by \autoref{lem:structural}), $T_c^1 =
\{ t \in T : x_{c,t} = 1 \}$. The last equality follows as sets $\F{c,t}$ are
disjoint for different $t$.

Within a single update operation, let $\Delta x_{c,t}$ and $\Delta y_f$ be the
increases of variables~$x_{c,t}$ and $y_f$, respectively. By
\autoref{lem:structural}, \FRAC increases one connection variable $x_{c,t^*}$
for an~active distance $t^*$ (and no connection variable if there is no active
distance) and performs augmentations of $y_f$ for all $f \in F_1$. The increase
of the primal value is then
\begin{align*}
\Delta \primal 
& = \sum_{t \in T} t \cdot \Delta x_{c,t} 
    + \sum_{f \in F_1} \cost(f) \cdot \Delta y_f 
 \leq 1 + \sum_{f \in F_1} \cost(f) 
    \cdot \left ( \frac{y_f}{\cost(f)} + \frac{1}{|F| \cdot \cost(f)}\right ) \\ 
& = 1 + y(F_1) + \frac{|F_1|}{|F|} < 3 ,
\end{align*}
where the last inequality follows by \eqref{eq:y_f_before_update}.
\end{proof}

\newcommand{\fracperformancetext}{%
For any input $(G = (F,C,E,\cost),A)$, it holds that $\FRAC(G,A) \leq O(\log
|F|) \cdot \OPT(G,A)$.}

\begin{lemma}
\label{lem:frac}
\fracperformancetext
\end{lemma}

\begin{proof}
Let $k$ be the total number of active clients in $A$, and let
$\mathsf{val}(\primal_k)$ and $\mathsf{val}(\dual_k)$ be the values of the final
primal and dual solutions generated by \FRAC. Then,
\begin{align*}
\FRAC(G,A) = \mathsf{val}(\primal_k) 
    & \leq 3 \cdot \mathsf{val}(\dual_k)
        && \text{(by \autoref{lem:primal_dual_relation})} \\
    & \leq O(\log |F|) \cdot \OPT(G,A)
        && \text{(by \autoref{lem:dual_almost_feasible} and weak duality).}
\qedhere
\end{align*}
\end{proof}

\section{Deterministic Rounding}
\label{sec:rounding}

Now we define our deterministic algorithm \INT, which rounds the fractional
solution computed by $\FRAC$. For a client $c \in A$, \INT observes the actions
of $\FRAC$ while processing $c$ and on this basis makes its own decisions.
First, \INT processes augmentations of variables $y_f$ performed by \FRAC, and
purchases some facilities. Once $\FRAC$ finishes handling client $c$, \INT
connects~$c$ to the closest open facility. (We show below that such facility
exists.)


\subsection{Purchasing Facilities: Properties of INTFAC}

Purchasing facilities by \INT is based solely on graph $G$ and on updates of
variables $y_f$ produced by \FRAC. In particular, it neglects whether a given
client is active or not. We use integral variables $\hy_f \in \{0,1\}$ to denote
whether \INT opened facility $f$. Furthermore, for any set $F'$ we use $\hy(F')$
as a shorthand for $\sum_{f \in F'} \hy_f$. 

The following lemma is an adaptation of the deterministic rounding routine for
the set cover problem by Alon et al.~\cite{AlAABN09} and its proof is postponed to 
\autoref{sec:intfac}.

\begin{lemma}
\label{lem:facility_rounding}
Fix any input $(G = (F,C,E,\cost),A)$. Initially, $\hy_f = y_f = 0$ for any $f
\in F$. There exists a~deterministic polynomial-time online algorithm \INTFAC
that transforms increments of fractional variables~$y_f$ to increments of
integral variables $\hy_f \in \{0,1\}$, so that
\begin{itemize}
\item condition $y(S_{c,t}) \geq 1/2$ implies $\hy(S_{c,t}) \geq 1$ for any
client $c \in C$ (active or inactive) and any $t \in T$, 
\item $\sum_{f \in F} \cost(f) \cdot \hy_f \leq O(\log |C \times T|) \cdot
		\sum_{f \in F} \cost(f) \cdot y_f + 2 \cdot \max_{f \in F} \cost(f)$.
\end{itemize}
\end{lemma}


\subsection{Connecting Clients}

Once \INT purchases facilities using deterministic routine \INTFAC
(cf.~\autoref{lem:facility_rounding}), it connects client $c$ to the closest
open facility. Now we show that such a facility indeed exists and we bound the
competitive ratio of \INT.

\begin{lemma}
\label{lem:det_connection}
On any input $(G,A)$, the solution generated by \INT is feasible and 
the total cost of connecting clients by \INT is at most $2 \cdot \FRAC(G,A)$.
\end{lemma}

\begin{proof}
Fix any client $c \in A$. By \autoref{lem:good_distance}, there exists
a distance $\tau \in T$ such that $y(S_{c,\tau}) \geq 1/2$ and $\sum_{t \in T} t
\cdot x_{c,t}  \geq \tau/2$. By \autoref{lem:facility_rounding}, once \INT
purchases facilities, it holds that $\hy(S_{c,\tau}) \geq 1$. It means that at
least one facility is opened in set $S_{c,\tau}$, i.e., at distance at most
$\tau$ from $c$. 

Therefore, \INT is feasible and by connecting client $c$ to the closest open facility,
it ensures that the connection cost is at most $\tau \leq 2 \cdot \sum_{t \in T}
t \cdot x_{c,t}$. The proof is concluded by observing that $\sum_{t \in T} t
\cdot x_{c,t}$ is the connection cost of \FRAC that can be attributed solely to
the connection of client $c$.
\end{proof}

\begin{lemma}
\label{lem:int_ratio}
For any input $(G = (F,C,E,\cost),A)$, it holds that  
$\INT(G,A) \leq q \cdot \log |F| \cdot (\log |C| + \log \log \Delta_G) \cdot
\OPT(G,A) + 2 \cdot \max_{f \in F} \cost(f)$, where $q$ is a universal constant
not depending on $G$ or $A$. Furthermore, \INT runs in time polynomial in $|G|$,
$|A|$, $\max_{e \in E} \cost(e)$, and $\max_{f \in F} \cost(f)$.
\end{lemma}

\begin{proof}
Let $\maxcost = \max_{f \in F} \cost(f)$. Then,
\begin{align*}
\INT(G,A) 
	& \leq \sum_{f \in F} \cost(f) \cdot \hy_f + 2 \cdot \FRAC(G,A) 
		&& \text{(by \autoref{lem:det_connection})} \\
	& \leq O(\log |C \times T|) \cdot \FRAC(G,A) + 2 \cdot \maxcost
		&& \text{(by \autoref{lem:facility_rounding})} \\
	& = O((\log |C| + \log |T|) \cdot \log |F|) \cdot \OPT(G,A) + 2 \cdot \maxcost
		&& \text{(by \autoref{lem:frac}).}
\end{align*}
The bound on the cost of \INT is concluded by using $|T| \leq 2 + \log
\Delta_G$.

By \autoref{lem:primal_is_feasible}, \FRAC running time is  
$\poly(|G|, |A|, \max_{e \in E} \cost(e), \max_{f \in F} \cost(f))$. On top of
that, \INT adds its own computations (in particular the rounding scheme of
\INTFAC), whose runtime is polynomial in $|G|$ and $|A|$. This implies the
second part of the lemma (the running time of~\INT).
\end{proof}


\subsection{Purchasing Facilities: Algorithm INTFAC}
\label{sec:intfac}

We start with a technical claim and later we define our rounding procedure \INTFAC.

\begin{lemma}
\label{lem:technical}
Fix any $q \in [0,1/2]$ and any $r \geq 0$. Let $X$ be a binary variable being
$0$ with probability $p > 0$. Then, $\E[\exp(q \cdot X)] \leq \exp(- (3/2) \cdot
q \cdot \ln p)$.
\end{lemma}
	
\begin{proof}
Using the definition of $X$, we have
\begin{align*}
		\E[\exp(q \cdot X)] 
		& = p \cdot \e^0 + (1 - p) \cdot \e^q 
		 = \exp(\ln p) + (1 - \exp(\ln p)) \cdot \e^q \\
		& \leq 1 + \ln p - e^q \cdot \ln p 
		 = 1 - \ln p \cdot (\e^q - 1) \\
		& \leq 1 - (3/2) \cdot q \cdot \ln p  \\
		& \leq \exp(-(3/2) \cdot q \cdot \ln p).
\end{align*}
In the first inequality, we used that $e^{x} \cdot 1 + (1-e^{x}) \cdot z \leq
(1+x) \cdot 1 + (-x) \cdot z$ for any $x \leq 0$ and $z \geq 1$ and in the
second one, we used that $e^{x} - 1 \leq 3x/2$ for any $x \in [0,1/2]$. 
\end{proof}


\subparagraph*{Algorithm Description.}

As we mentioned earlier, our routine \INTFAC for rounding facilities is an
adaptation of the deterministic rounding procedure for the set cover problem by
Alon et al.~\cite{AlAABN09}. On the basis of the facility-client graph $G$, we
define the set $C \times T$ of \emph{elements}. Intuitively, our solution \FRAC
``covers'' an element $(c,t) \in C \times T$ by fractionally opening facilities
from $S_{c,t}$. The routine \INTFAC deterministically rounds these covering
choices.

Let $\ell = |C \times T|$, 
$\maxcost = \max_{f \in F} \cost(F)$ and $b = 6 \cdot \ln \ell = O(\log |C \times T|)$.
We consider the potential function $\Phi = \Phi_1 + \Phi_2$, where 
\begin{align*}
	\Phi_1 
		= \sum_{(c,t) \,:\, \hy(S_{c,t}) = 0} \ell^{\,4 \cdot y(S_{c,t})}
	&& \text{and} &&
	\Phi_2
		= \ell \cdot \exp\left( 
				\sum_{f\in F} \frac{\cost(f)}{2 \maxcost} \cdot \left(\hy_f - b \cdot y_f\right)
			\right).
\end{align*}

Assume that \FRAC augmented variable $y_f$. Then our algorithm \INTFAC chooses
whether to set $\hy_f$ to $1$ or not (purchase~$f$ or not), so that the
potential $\Phi$ does not increase. (We again emphasize that this choice
neglects the current set of active clients.)


\subparagraph*{Correctness and Performance.}

In the lemma below, we show that \INTFAC is well defined, i.e., it is possible
to fix variable $\hy_f$, so that the potential $\Phi$ does not increase. This
implies that both $\Phi_1$ and $\Phi_2$ remain upper-bounded, which can be in turn
used to show properties of \autoref{lem:facility_rounding}.

\begin{lemma}
\label{lem:potential_drops}
Assume $y_{f^*}$ is increased by $\delta$. If $\hy_{f^*} = 1$, then $\Phi$ does
not increase. Otherwise, there is a~choice to either set $\hy_{f^*}$ to $1$ or
not, so that $\Phi$ does not increase.
\end{lemma}

\begin{proof}
By $y_f$ and $\hy_f$, we mean the values of these variables before an update
operation of \FRAC.

First, we assume $\hy_{f^*} = 1$. Increasing variable $y_{f^*}$ affects values
of $y(S_{c,t})$ for $f^* \in S_{c,t}$: all such $y(S_{c,t})$ increase by
$\delta$. However, for any element $(c,t)$, such that $f^* \in S_{c,t}$, it
holds that $\hy(S_{c,t}) \geq \hy_{f*} = 1$, i.e., element $(c,t)$ is not
counted in the sum occurring in $\Phi_1$. Thus, increasing variable
$y_{f^*}$ does not affect $\Phi_1$. Furthermore, increasing $y_{f^*}$ and
keeping $\hy_{f^*}$ unchanged can only decrease $\Phi_2$. Thus, $\Phi = \Phi_1 +
\Phi_2$ does not increase when  $\hy_{f^*} = 1$. 

Second, we consider the case $\hy_{f^*} = 0$. To show that either setting
$\hy_{f^*}$ to $1$ or leaving it at $0$ does not increase the potential, we use
the probabilistic method and show that if we pick such action randomly (setting
$\hy_{f^*} = 1$ with probability $1-\ell^{-4 \cdot \delta}$), then, in
expectation, neither $\Phi_1$ nor $\Phi_2$ increases.

\begin{itemize}
\item 
As observed above, only elements $(c,t)$ for which $S_{c,t}$ contain $f^*$ are
affected by the increase of $y_{f^*}$ and possible change of $\hy_{f^*}$. Let $Q
= \{ (c,t) : f^* \in S_{c,t} \text{ and } \hy(S_{c,t}) = 0 \}$ be the set of
such elements contributing to $\Phi_1$. 

Fix any element $(c,t) \in Q$. Its initial contribution towards $\Phi_1$ is
$\ell^{\,4 \cdot y(S_{c,t})}$ and when $y_{f^*}$ increases, the contribution
grows to $\ell^{\,4 \cdot (y(S_{c,t}) + \delta)}$. However, with probability
$1- \ell^{-4 \cdot \delta}$, variable $\hy_{f^*}$ is set to $1$, thus
$\hy(S_{c,t})$ grows from $0$ to $1$, and in effect element $(c,t)$ stops
contributing to $\Phi_1$. Hence, the expected final contribution of element
$(c,t)$ towards $\Phi_1$ is $\ell^{\,4 \cdot (y(S_{c,t}) + \delta)} \cdot
\ell^{-4 \cdot \delta} + 0 \cdot (1 -  \ell^{-4 \cdot \delta}) = \ell^{\,4 \cdot
y(S_{c,t})}$, i.e., is equal to its initial contribution. Therefore, in
expectation, the value of $\Phi_1$ is unchanged. 

\item 
It remains to bound the expected value of $\Phi_2$. 
 Let $\hY$ be the random variable equal to the value of
$\hy_{f^*}$ after the random choice (i.e., $\hY = 1$ with probability $1-\ell^{-4
\cdot \delta}$) and $\Phi'_2$ denote the value of $\Phi_2$ after increasing
$y_{f^*}$ and after the random choice. Using $y_{f^*} = 0$, we obtain
\begin{align*}
	\Phi'_2
	& = \ell \cdot \exp\left(
		\sum_{f\in F} \frac{\cost(f)}{2 \maxcost} \cdot \left(\hy_f - b \cdot y_f\right)
		+ \frac{\cost(f^*)}{2\maxcost} \cdot \hY
		- \frac{b \cdot \cost(f^*)}{2 \maxcost} \cdot \delta
	\right) \\
	& = \Phi_2 \cdot \exp\left(
			\frac{\cost(f^*)}{2\maxcost} \cdot \hY 
		\right) \cdot \exp\left(
		- \frac{b \cdot \cost(f^*)}{2 \maxcost} \cdot \delta
	\right).
\end{align*}
To estimate $\E[\Phi'_2]$, we upper-bound the expected value of expression
$\exp( \hY \cdot \cost(f^*) / (2\maxcost))$, using \autoref{lem:technical}
with $q = \cost(f^*) / (2\maxcost) \leq 1/2$ and $p = \ell^{-4 \cdot \delta}$,
obtaining that 
\[
	\E\left[\exp\left(
			\frac{\cost(f^*)}{2\maxcost} \cdot \hY
		\right)\right]
	\leq \exp\left(- \frac{(3/2) \cdot \cost(f^*)}{2\maxcost} 
		\cdot \ln p \right) 
	= \exp\left(\frac{6 \cdot \ln \ell \cdot \cost(f^*)}{2\maxcost} \cdot \delta 
		\right) .
\]
Therefore, $\E[\Phi'_2] \leq \Phi_2$ and the lemma follows.
\qedhere
\end{itemize}
\end{proof}

\begin{proof}[Proof of \autoref{lem:facility_rounding}]
Initially, all variables $y_f$ and $\hy_f$ are zero, and thus $\Phi =
\sum_{(c,t) \in C \times T} \ell^0 + \ell \cdot \exp(0) = 2 \cdot \ell$. By
\autoref{lem:potential_drops}, the potential never increases. Since $\Phi_2$
is non-negative, any summand of $\Phi_1$ is always at most $2 \cdot \ell \leq
\ell^2$. Therefore, $4 \cdot y(S_{c,t}) \geq 2$ always implies $\hy(S_{c,t}) >
0$, i.e., the first part of the lemma follows. 

To show the second part, we again use that $\Phi = \Phi_1 + \Phi_2 \leq 2 \cdot
\ell$ at any time. As $\Phi_1$ is non-negative, $\Phi_2 \leq 2 \cdot \ell$.
Substituting the definition of $\Phi_2$, dividing by $\ell$, and taking natural
logarithm of both sides yields 
\[
	\frac{1}{2 \maxcost} \cdot \sum_{f\in F}\left(\hy_f \cdot \cost(f) -
		b \cdot y_f \cdot \cost(f) 
	\right) \leq \ln(2) < 1.
\]
Therefore, $
	\sum_{f\in F} \hy_f \cdot \cost(f) \leq 2 \maxcost + 
	b \cdot \sum_{f\in F} y_f \cdot \cost(f)$.
\end{proof}

\section{Handling Large Aspect Ratios}
\label{sec:doubling}

The guarantee of \autoref{lem:int_ratio} has two deficiencies: (i) the bound
on the competitive ratio of \INT depends on the aspect ratio of $G$ and on the
cost of the most expensive facility, (ii) the running time of \INT depends on
the maximal cost in graph $G$ (which can be exponentially large in the input
description). We show how to use cost doubling and edge pruning to handle these
issues, creating our final deterministic solution \DET and proving the main theorem
(restated below).

\maintheorem*

\begin{proof}
Fix facility-client graph $G = (F,C,E,\cost)$ for the non-metric facility
location problem. Recall that we assumed that all non-zero costs and distances
in $G$ are powers of $2$ and are at least~$1$. Let $R = \log |F| \cdot (\log |C|
+ \log \log (|F| \cdot |C|))$.

We now construct a deterministic algorithm \DET which is $O(R)$-competitive on
an~input~$(G,A)$. Let $q$ be the constant from \autoref{lem:int_ratio}. \DET
operates in phases, numbered from~$0$. In phase~$j$, it executes the following
operations.

\begin{enumerate}

\item \DET \emph{pre-purchases} all facilities and edges of $G$ whose cost is
smaller than $2^j / (|F| \cdot |C|)$. 

\item \DET creates an auxiliary facility-client graph $\tG_j$ applying the
following modifications to~$G$.
\begin{itemize}
\item First, \DET creates graph $G_j$ containing only edges and facilities from
$G$ whose individual cost is at most~$2^j$. It also removes connections to
facilities that have been removed in this process.

\item Second, the costs of all facilities and edges that have been pre-purchased
by \DET are set to zero in $G_j$. In a result, $G_j$ is a sub-graph of $G$ with
adjusted distances and costs of facilities, has the same set of clients, its set
of facilities is a subset of $F$, and $\Delta_{G_j} \leq |F| \cdot |C|$. 

\item Third, $\tG_j$ is the modified version of $G_j$, where all costs have been
scaled down, so that the smallest positive cost is equal to $1$. We denote the
scaling factor by $h_j \leq 1$. 
\end{itemize}

\item \DET simulates algorithm $\INT$ on input $(\tG_j,A)$. That is, for a
client $c \in A$, \DET verifies whether the overall cost of $\INT$ (including
serving $c$) remains at most $h_j \cdot (q \cdot R + 2) \cdot 2^j$. In such
case, \DET outputs the choices of \INT for client $c$ as its own. We emphasize
that \INT is run also on clients that have been already served in the previous
phases; in effect, \DET may purchase the same facilities or connections multiple
times. 

\item Eventually, either the sequence $A$ of active clients ends and the total
cost of $\INT$ on $(\tG_j,A)$ is at most $h_j \cdot (q \cdot R + 2) \cdot 2^j$
(in which case \DET terminates as well) or the purchases made by $\INT$, while
handling a client $c \in A$, caused its cost to exceed $h_j \cdot (q \cdot R +
2) \cdot 2^j$. (This includes the special case where $c$ is disconnected from
all facilities in $\tG_j$, because all edges incident to $c$ in $G$ were either
more expensive than $2^j$ or were leading to facilities more expensive than
$2^j$.) In the case of exceeded cost, \DET disregards the decisions of $\INT$
for client $c$, terminates $\INT$, and starts phase $j+1$, processing also all
clients that were already served in phase $j$.
\end{enumerate}

We now analyze the performance of \DET. Let $k = \lceil \log (\OPT(G,A)) \rceil
\geq 0$. We show that \DET terminates latest in phase $k$. Assume that \DET has
not finished within phases $0, 1, \ldots, k-1$. In phase $k$, \DET creates
auxiliary graphs $G_k$ and $\tG_k$, and runs $\INT$ on graph $\tG_k$. Graph
$G_k$ contains all edges of $G$ of cost at most $2^k$; their cost in $G_k$ is
the same or reset to zero. As $\OPT(G,A) \leq 2^k$, $\OPT(G,A)$ purchases only
edges that are in $G_k$, and thus $\OPT(G,A)$ is also a feasible solution to
instance $(G_k,A)$. Thus, $\OPT(G_k,A) \leq \OPT(G,A) \leq 2^k$. As $\tG_k$ is
the scaled-down copy of $G_k$, $\OPT(\tG_k,A) = h_k \cdot \OPT(G_k,A) \leq h_k
\cdot 2^k$.

Let $\tilde{F}_k$ be the set of facilities of graph $\tG_k$ and
$\tilde{\cost}_k(f)$ is the cost of opening facility $f$ in graph $\tG_k$.
Clearly, $|\tilde{F}_k| \leq |F|$ and $\tilde{\cost}_k(f) \leq h_k \cdot
\cost(f)$ for any $f \in F$. By our construction, $\Delta_{\tG_k} = \Delta_{G_k}
\leq |F| \cdot |C|$. Hence, \autoref{lem:int_ratio} implies that
\begin{align*}
    \INT(\tG_k,A) 
    & \leq q \cdot \log |F_k| \cdot \left(\log |C| + \log \log \Delta_{\tG_k}\right) 
        \cdot \OPT(\tG_k,A) + 2 \cdot \max_{f \in \tilde{F}_k} \tilde{\cost}_k(f) \\
    & \leq h_k \cdot q \cdot \log |F| \cdot (\log |C| + \log \log (|F| \cdot |C|)) 
        \cdot 2^k + 2 \cdot h_k \cdot 2^k \\
    & = h_k \cdot (q \cdot R + 2) \cdot 2^k .
\end{align*}
Therefore, \INT is not terminated prematurely within phase $k$ because of high
cost and it finishes the entire sequence $A$. This implies the feasibility of \INT:
it serves all clients latest in phase $k$. 

To bound the total cost of \DET, recall that at the beginning of phase $j$, \DET
purchases at most $|F| \cdot |C|$ edges and at most $|F|$ facilities, each of
cost at most $2^j / (|F| \cdot |C|)$. The associated overall cost is at most $2
\cdot 2^j$. The cost of the subsequent execution of algorithm $\INT$ on $\tG_j$
is, by our termination rule, at most $h_j \cdot (q \cdot R + 2) \cdot 2^j$, and
thus the cost incurred by repeating $\INT$'s actions on $G$ is at most $(q \cdot
R + 2) \cdot 2^j$. The overall cost is then $\DET(G,A) \leq \sum_{j=0}^k (q
\cdot R + 4) \cdot 2^j = O(R) \cdot 2^k = O(R) \cdot \OPT(G,A) = O(\log |F|
\cdot (\log |C| + \log \log |F|)) \cdot \OPT(G,A)$.

For the running time of \DET, we note that in phase $j$, $\INT$ is run on a
graph $\tG_j$ whose smallest cost is $1$, and hence the largest cost is at
most $\Delta_{\tG_j} = \Delta_{G_j} \leq |F| \cdot |C|$. Thus, by
\autoref{lem:int_ratio}, the running time of $\INT$ in a single phase is
polynomial in $|G|$ and $|A|$, and the number of phases is logarithmic in the
maximum cost occurring in $G$, and thus also polynomial in $|G|$.
\end{proof}


\section{Application to Online Node-Weighted Steiner Tree}
\label{sec:steiner}

Our result for the non-metric FL problem has an immediate application for the
online node-weighted Steiner tree (NWST) problem, where the graph consists of
$\ell$ nodes and an online algorithm is given $k$ terminals to be connected.
Namely, the randomized solution for the online NWST problem by Naor et
al.~\cite{NaPaSi11} is in fact a deterministic polynomial-time ``wrapper''
around randomized routine solving the non-metric FL problem. To solve an
instance of the NWST problem, their algorithm constructs a sub-instance of
non-metric FL with $O(\ell)$ facilities, $O(\ell)$~potential clients, and $O(k)$
active clients. Such instance can be solved by the randomized algorithm of Alon
et~al.~\cite{AlAABN06} with the competitive ratio of $O(\log k \cdot \log
\ell)$. The wrapper adds another $O(\log k)$ factor in the ratio, resulting in
an $O(\log^2 k \cdot \log \ell)$-competitive algorithm. 

Our deterministic algorithm, when applied to this setting would be $O(\log^2
\ell)$-competitive on the constructed non-metric FL sub-instance. Therefore, by
replacing the randomized algorithm by Alon et al.~\cite{AlAABN06} with our
deterministic one, we immediately obtain the first online deterministic solution
for online NWST. 

\begin{corollary}
\label{cor:steiner_tree}
There exists a polynomial-time deterministic online algorithm for the
node-weighted Steiner tree problem, which is $O(\log k \cdot \log^2
\ell)$-competitive on graphs with $\ell$ nodes and $k$ terminals.
\end{corollary}

We note that the currently best solution for the node-weighted Steiner tree is
randomized and achieves the ratio of $O(\log^2 \ell)$~\cite{HaLiPa17,HaLiPa14}
and the best known lower bound for deterministic algorithms is $\Omega(\log \ell
\cdot \log k / (\log \log \ell + \log \log k))$~\cite{NaPaSi11,AlAABN09}.
\section{Final Remarks}

We presented a deterministic solution to the non-metric facility location
problem, whose performance nearly matches that of the best randomized one. By
clustering facilities, we encoded dependencies between facilities and clients,
which allowed us later to apply the rounding scheme to facilities only, neglecting the
actual active clients. It would be however interesting and useful to have an
online deterministic rounding routine able to handle such dependencies
internally (e.g., by creating a pessimistic estimator that can be computed and
handled in an online manner), as it is the case for the set cover problem or
throughput-competitive virtual circuit routing~\cite{BucNao09b}.

That said, we believe that our distance clustering techniques can be extended to
other network design problems for which only randomized algorithms existed so
far, e.g., online multicast problems on trees~\cite{AlAABN06}, online group
Steiner problem on trees~\cite{AlAABN06}, or variants of the facility location
problem that are used as building blocks for solutions to other node-weighted
Steiner problems~\cite{HaLiPa14,HaLiPa17}. (For these problems there are no
known direct reductions to the set cover problem). Finally, another open problem
is whether these techniques could be also applied more directly for the
node-weighted Steiner tree, resulting in a better deterministic competitive
ratio.

\section*{Acknowledgments}

We thank Marek Adamczyk and Christine Markarian for helpful discussions. 
We thank anonymous reviewers of an earlier draft for pointing us to the
reduction of Kolen and Tamir~\cite{KolTam90}.

\bibliographystyle{plainurl}
\bibliography{ref}

\end{document}